%% file: main.tex
\documentclass[journal,12pt,onecolumn,draftclsnofoot,]{IEEEtran}

\usepackage{tikz}
\usetikzlibrary{automata,positioning}
\usepackage[T1]{fontenc}
\usepackage{amsmath,amsfonts,amssymb,amsthm,amsbsy}
\usepackage{float}
\usepackage{textcomp}
\usepackage{cite}
\usepackage{bm}
\usepackage{algorithm} 
\usepackage{xcolor}
\usepackage{arydshln}
\usepackage{multicol}
\makeatletter
\newsavebox\myboxA
\newsavebox\myboxB
\newlength\mylenA
\usepackage{balance}
\newcommand*\xoverline[2][0.70]{%
    \sbox{\myboxA}{$\m@th#2$}%
    \setbox\myboxB\null
    \ht\myboxB=\ht\myboxA%
    \dp\myboxB=\dp\myboxA%
    \wd\myboxB=#1\wd\myboxA
    \sbox\myboxB{$\m@th\overline{\copy\myboxB}$}
    \setlength\mylenA{\the\wd\myboxA}
    \addtolength\mylenA{-\the\wd\myboxB}%
    \ifdim\wd\myboxB<\wd\myboxA%
       \rlap{\hskip 0.5\mylenA\usebox\myboxB}{\usebox\myboxA}%
    \else
        \hskip -0.5\mylenA\rlap{\usebox\myboxA}{\hskip 0.5\mylenA\usebox\myboxB}%
    \fi}
\makeatother
\usepackage[algo2e]{algorithm2e} 
\usepackage[noend]{algpseudocode}

\usepackage{lipsum}
\usepackage{authblk}

\makeatletter
\renewcommand\AB@affilsepx{ \protect\Affilfont} 
\makeatother
\makeatletter
\def\BState{\State\hskip-\ALG@thistlm}
\makeatother
\usepackage{mathtools}

\theoremstyle{remark}

\usepackage[section]{placeins}
\usepackage{amsmath}

\usepackage[cmintegrals]{newtxmath}

\newtheorem{theorem}{\textbf{Theorem}}
\renewenvironment{proof}{{\textbf{Proof.}}}{}
\begin{document}

	\title{Diversity Analysis of Millimeter-Wave OFDM Massive MIMO Systems}

  	\small
 	\author{Sadjad Sedighi, \textit{Student Member, IEEE,}}
 	\author{Ender Ayanoglu, \textit{Fellow, IEEE}}
 	\affil{CPCC, Dept of EECS, UC Irvine, Irvine, CA, USA,}
 	\normalsize

 	\maketitle

\input{abstract}
\input{introduction}
\input{model}
\input{diversity}

\input{simulation}
\input{conclusion}

 	\bibliographystyle{IEEEtran}
 	\bibliography{References}
	\end{document}

%% file: abstract.tex
	\begin{abstract}
We analyze the diversity gain for a distributed antenna subarray employing orthogonal frequency-division multiplexing (OFDM) in millimeter-wave (mm-Wave) massive multiple-input multiple-output (MIMO) systems. We show that the diversity gain depends on the number of transmitted data streams, the number of remote antenna units, and the number of propagation paths between RAUs. Furthermore, we show that by using bit-interleaved coded multiple beamforming (BICMB), one can achieve the maximum diversity gain in a distributed antenna subarray system. The assumption in both scenarios is that the number of the antennas at the transmitter and the receiver are large enough and channel state information (CSI) is known at the transmitter and the receiver. 
	\end{abstract}

%% file: introduction.tex
\section{Introduction} \label{sec:introduction}

It has been shown that distributed antenna subarray architecture can be suitable for radio communications where it increases the spectral efficiency. Furthermore, it can compensate the path loss in millimeter-wave (mm-Wave) frequency bands by using multiple remote antenna units (RAUs) \cite{Clark2001,Roh2002,Dai2011, Wang2013,Qing2015, Juan2018}. 

Diversity gain in such a system was first studied in \cite{Dian2018J}, where the authors derived a diversity gain formula for a flat-fading channel model for both single-user and multi-user scenarios. The diversity gain in both scenarios depends on the number of RAUs at both the transmitter and the receiver, the number of transmitted data streams, and the number of propagation paths in each subchannel. In other words, by increasing the number of transmitted data streams, the diversity gain decreases. 

In \cite{Sedighi2020J}, we showed that by using bit-interleaved coded multiple beamforming (BICMB), one can achieve the maximum diversity gain mm-Wave massive MIMO system with a distributed antenna subarrays. Also, in \cite{Sedighi2020L}, we used BICMB with perfect coding (BICMB-PC) to achieve the maximum diversity gain. The diversity gain derived in \cite{Sedighi2020J,Sedighi2020L} is based on a flat-fading channel model. Also, the diversity gain is independent of the number of transmitted data streams and depends on the number of RAUs at both the transmitter and the receiver, the number of propagation paths, and the large-scale fading coefficient in each subchannel. Authors in \cite{Xiao2015} used the iterative eigenvalue decomposition (EVD) to find the antenna weight vectors to achieve  full diversity gain for mm-Wave massive MIMO systems. In \cite{Elganimi2018}, a combination of space-time block coded spatial modulation with hybrid analog-digital beamforming was used to achieve the full diversity gain for mm-Wave MIMO systems.

Most of the works studying the mm-Wave massive MIMO systems with a distributed antenna subarray are based on the flat-fading channel model, unlike practical applications where the channel is frequency-selective \cite{Pi2011}. One can use orthogonal frequency division multiplexing (OFDM) modulation to overcome this problem and divide the overall channel into multiple flat-fading channels. Beamforming design for a mm-Wave OFDM distributed antenna systems is studied in \cite{Zhang2020}, where the authors propose a cooperative wideband hybrid beamforming method under the transmitting power constraints of each RAU. Also, authors in \cite{Sohrabi2017} find the optimum beamformers for a single-user the OFDM system which employs hybrid beamforming.

In this paper, we first analyze the diversity gain for mm-Wave OFDM massive MIMO system with distributed antenna subarray. Furthermore, by employing BICMB, we prove that the maximum diversity gain can be achieved for mm-Wave OFDM massive MIMO system with distributed antenna subarray. Employing BICMB to achieve the maximum diversity gain in MIMO OFDM systems was first studied in \cite{Akay2007}. Later on, authors in \cite{Li2013} used more general assumptions to achieve the maximum diversity gain in a MIMO OFDM system.

%% file: model.tex
\section{System Model} \label{sec:model}

Consider a wideband mm-Wave distributed antenna system as in Fig. \ref{fig::bicmb_ofdm_model}
consisting of $M_t$ RAUs at the transmitter where each RAU consists of $N_t$ antennas and $M_r$ RAUs at the receiver where each RAU has $N_r$ antennas. Base station (BS) transmits $N_s$ data symbols per frequency tone. The number of data symbols can be different for different frequency tones. However, for the sake of simplicity, the focus of this letter is only on the case with an equal number of data streams for all subcarriers. Since there are tens of subcarriers in a typical OFDM system, the analysis of the case where different subcarriers have different number of data streams would be prohibitively complex. Furthermore, for mm-Wave systems with highly correlated channels, all the subchannels are typically low rank \cite{Sohrabi2017}. Furthermore, since implementing  fully digital beamforming is expensive, we consider a hybrid beamforming architecture, where the number of RF chains is much smaller than the number of antennas at both the transmitter side and the receiver side, i.e., $N_t^{\text{RF}} \ll N_t$ and $N_r^{\text{RF}} \ll N_r$. The overall beamformer in the hybrid beamforming architecture consists of a low-dimensional digital beamformer and a high-dimensional analog beamformer.

\ifCLASSOPTIONonecolumn
 \begin{figure*}[!t]
         \centering
         \includegraphics[width=\textwidth]{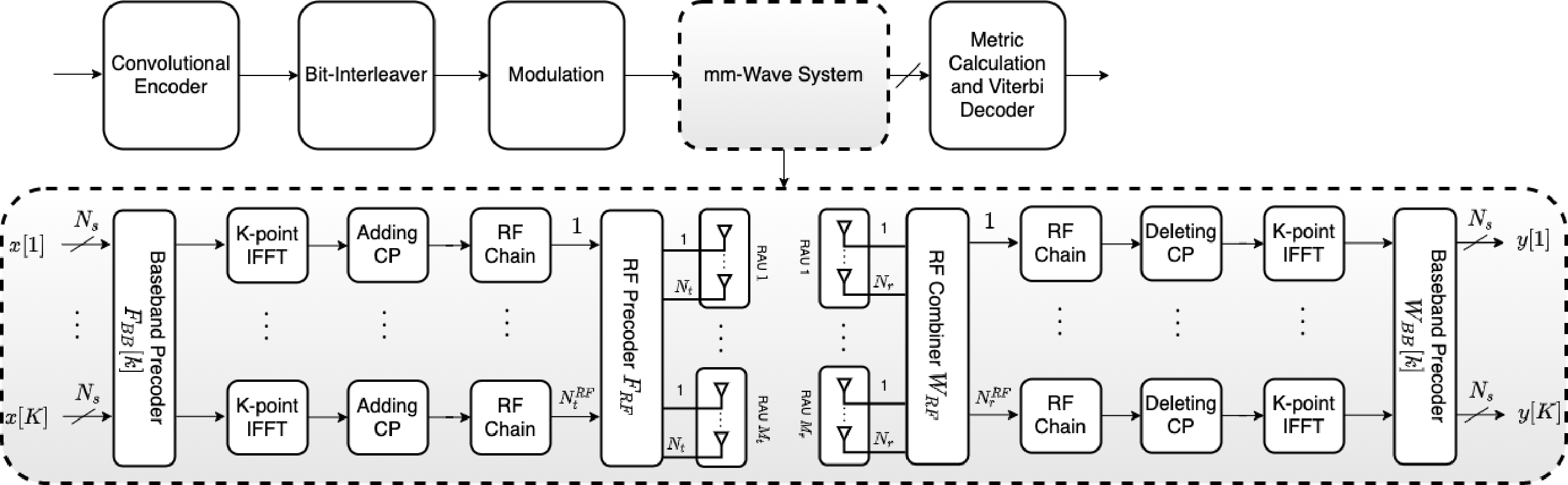}
 		\caption{Block diagram of a single-user mm-Wave massive MIMO system with distributed antenna subarray employing OFDM.}
 		\label{fig::bicmb_ofdm_model}
 	\end{figure*}
 \else
  \begin{figure}[!t]
 	\centering
 	\includegraphics[width=0.48\textwidth]{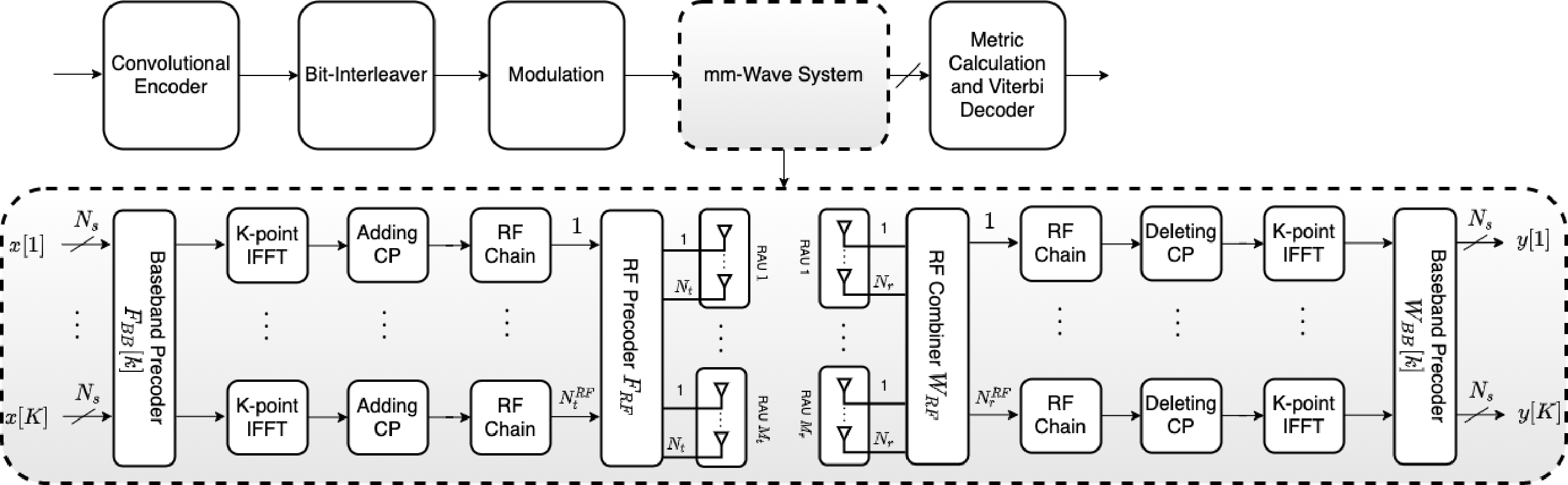}
 	\caption{Block diagram of a single-user mm-Wave massive MIMO system with distributed antenna subarray employing OFDM.}
 	\label{fig::bicmb_ofdm_model}
 \end{figure}
 \fi
Let $K $ denote the number of subcarriers in the system and $N_{cp}$ the number of cyclic prefix (CP) samples. Therefore, the OFDM symbol length and the CP length are $T_k = KT_s$ and $T_{cp} = N_{cp}T_s$, respectively, where $T_s$ is the system sampling interval.

One can write the channel response in the frequency domain as
	\begin{align} \label{eq::bicmb_ofdm_su_channel} 
	\mathbf{H}[k]=
	\begin{bmatrix}
    \mathbf{H}_{1,1}[k] & \mathbf{H}_{1,2}[k]     & \dots  & \mathbf{H}_{1,M_t}[k] \\
	\mathbf{H}_{2,1}[k] & \mathbf{H}_{2,2} [k]    & \dots  & \mathbf{H}_{2,M_t} [k] \\
	\vdots           & \vdots               & \ddots & \vdots              \\
    \mathbf{H}_{M_r,1} [k] & \mathbf{H}_{M_r,2} [k] & \dots  & \mathbf{H}_{M_r,M_t} [k]
	\end{bmatrix},
	\end{align}
where $\mathbf{H}_{i,j}[k]$ is the subchannel between $i$th RAU at the receiver and $j$th RAU at the transmitter, at $k$th subcarrier, defined as
\begin{align}\label{eq::bicmb_ofdm_su_channel_freq}
    \mathbf{H}_{i,j} [k]= \sqrt{\beta\frac{N_rN_t}{L}} \sum_{l=1}^{L} \alpha_{i,j}^l \mathbf{a}_r(\theta_{i,j}^l) \mathbf{a}_t^H(\phi_{i,j}^l)e^{-j2\pi\frac{k}{K}\frac{\tau_{i,j}^l}{T_s}},
\end{align}
	in which $L$ is the number of propagation paths and $\alpha_{ij}^{l}$ is the complex gain of the $l$th ray which follows $\mathcal{CN}(0,1)$, $\theta_{ij}^l \in [0,2\pi]$,  $\phi_{ij}^l \in [0,2\pi]$, for all 
$i$, $j$, $l$, and the vectors $\mathbf{a}_r(\theta_{ij}^{l})$ and $ \mathbf{a}_t(\phi_{ij}^{l})$ are the normalized array response at the receiver and transmitter, respectively. In particular, this letter adopts a  uniform linear array (ULA) where both $\mathbf{a}_r(\theta_{ij}^{l})$ and $ \mathbf{a}_t(\phi_{ij}^{l})$ are modeled as
\begin{align}
	\mathbf{a}_{\text{ULA}}(\varphi)=\frac{1}{\sqrt{N}}\left[1, e^{j\frac{2\pi}{\lambda}d\sin(\varphi)},\dots, e^{j(N-1)\frac{2\pi}{\lambda}d\sin(\varphi)}\right]^T,
\end{align}
where $\lambda$ is the transmission wavelength and $d$ is the antenna spacing. Please note that in this letter the assumption is that the operating frequency is much larger than the total bandwidth of the OFDM system. This results in  signal wavelength in all subcarriers to be considered approximately equal.

In the mm-Wave hybrid beamforming architecture shown in Fig. \ref{fig::bicmb_ofdm_model}, the transmitter first precodes $N_s$ data symbols $\mathbf{x}[k]$ at each subcarrier $k=1,\dots,K$, using a low-dimensional digital precoder, $\mathbf{V}_\text{D}[k] \in \mathbb{C}^{N_{t}^\text{RF} \times N_s}$, then transforms the signals to the time domain by using $N_{t}^\text{RF}$ $K$-point inverse fast Fourier transforms (IFFTs). After adding cyclic prefixes, the transmitter employs an analog precoding matrix $\mathbf{V}_\text{RF} \in \mathbb{C}^{N_t \times N_{t}^\text{RF}}$, to generate the final transmitted signal. Since the analog precoding is done after the IFFT block, the analog precoder is identical for all subcarriers. This is the key challenge in designing the hybrid beamformers in OFDM systems as compared to single-carrier systems \cite{Sohrabi2017}.

After receiving the transmitted signals from all the RAUs and all subcarriers, the received signals are processed by using the analog beamformer $\mathbf{W}_{\text{RF}}$, where $\mathbf{W}_{\text{RF}}$ is an $M_rN_r\times N_{r}^{\text{RF}}$ complex matrix. Then, the  receiver removes the CP and then performs fast Fourier transform (FFT). Finally, for each subcarrier, a low-dimensional digital combiner is used to process the the frequency domain signals. The processed signal at the $k$th subcarrier can be written as 
\begin{align} \label{eq::bicmb_ofdm_rec_sig}
	\mathbf{y}[k] = \mathbf{W}_{\text{BB}}^H[k]  \mathbf{W}_{\text{RF}}^H \mathbf{H}[k] \mathbf{F}_{\text{RF}} \mathbf{F}_{\text{BB}}[k] \mathbf{x}[k] + \mathbf{W}_{\text{BB}}^H[k]  \mathbf{W}_{\text{RF}}^H \mathbf{n}[k],
\end{align}
where $\mathbf{n}[k] $  is the additive complex white  Gaussian noise with zero mean and variance $N_t/\text{SNR}$ for subcarrier $k$.

In a similar fashion to \cite{Ayach2012, Sedighi2020J}, one can find the beamforming matrices at both the transmitter and the receiver sides. We first rewrite (\ref{eq::bicmb_ofdm_su_channel}) in a more compact form, then by assuming that the number of antennas at both the transmitter and the receiver goes to infinity, we find the optimal solutions for the beamforming matrices \cite{Sohrabi2017}.

One can easily see from the subchannel and channel matrix definitions in (\ref{eq::bicmb_ofdm_su_channel}) and (\ref{eq::bicmb_ofdm_su_channel_freq}) that there is a total of $L_s=M_rM_tL$ propagation paths. We can write the channel matrix $\mathbf{H}[k]$  as a summation of $L_s$ matrices, where each matrix is rank one. In other words 
\begin{align} \label{eq::bicmb_su_channel_composition}
\mathbf{H}[k] = \sum_{i=1}^{M_r}\sum_{j=}^{M_t}\sum_{l=1}^{L_{ij}}\tilde{\alpha}^l_{ij}
\tilde{\mathbf{a}}_r(\theta_{ij}^{l}) \tilde{\mathbf{a}}_t^H(\phi_{ij}^{l}) e^{-j2\pi\frac{k}{K}\frac{\tau_{i,j}^l}{T_s}},
\end{align}
where
\begin{align}
\tilde{\alpha}^l_{ij}[k]=\sqrt{\beta\frac{N_rN_t}{L}} {\alpha}^l_{ij}e^{-j2\pi\frac{k}{K}\frac{\tau_{i,j}^l}{T_s}}.
\end{align}
Note that in (\ref{eq::bicmb_su_channel_composition}) we define $\tilde{\mathbf{a}}_r(\theta_{ij}^l)$ as an $M_rN_r\times 1$ vector whose $p$th entry is defined as
\begin{align}
[\tilde{a}_{r}(\theta_{ij}^{l})]_{p} =  \left\{
\begin{array}{ll}
[\tilde{a}_{r}(\theta_{ij}^{l})]_{p-(i-1)N_r} , & (i-1)N_r < p \leq iN_r  \\
0, & \text{otherwise}.
\end{array}
\right.
\end{align}
Similarly, we define a $M_tN_t\times 1$ vector $\tilde{\mathbf{a}}_t(\phi_{ij}^l)$ as 
\begin{align}
[\tilde{a}_{t}(\phi_{ij}^{l})]_{q} =  \left\{
\begin{array}{ll}
[\tilde{a}_{t}(\phi_{ij}^{l})]_{q-(i-1)N_t} , & (i-1)N_t < q \leq i N_t  \\
0, & \text{otherwise}.
\end{array}
\right.
\end{align}
We sort $\tilde{\alpha}_{ij}^l[k]$s in a descending way for each subcarrier such that $|\tilde{\alpha}_{1}[k]|\geq |\tilde{\alpha}_{2}[k]| \geq \dots \geq |\tilde{\alpha}_{L_s}[k]|$. We sort the corresponding $\theta_{ij}^l$s and $\phi_{ij}^l$s based on this sort. Then, one can  decompose the channel matrix defined in (\ref{eq::bicmb_su_channel_composition}) as 
\ifCLASSOPTIONonecolumn
\begin{align} \label{eq::bicmb_su_channel_decompose}
\mathbf{H}[k] &= \mathbf{A}_r \mathbf{D}[k] \mathbf{A}_t^* = \left[ \mathbf{A}_r |\mathbf{A}_r^{\perp}\right]\widetilde{\mathbf{D}}\left[\text{diag}([e^{j\angle \tilde{\alpha}_1},\dots,e^{j\angle \tilde{\alpha}_{L_s}}])\mathbf{A}_t|\mathbf{A}^{\perp}_t\right]^*,
\end{align}
\else
\begin{align} \label{eq::bicmb_su_channel_decompose}
\mathbf{H}[k] &= \mathbf{A}_r \mathbf{D}[k] \mathbf{A}_t^* \nonumber \\ &=\left[\mathbf{A}_r|\mathbf{A}_r^{\perp}\right]\widetilde{\mathbf{D}}\left[\text{diag}([e^{j\angle \tilde{\alpha}_1},\dots,e^{j\angle \tilde{\alpha}_{L_s}}])\mathbf{A}_t|\mathbf{A}^{\perp}_t\right]^*,
\end{align}
\fi
where we define the matrix $\mathbf{A}_r =[\tilde{\mathbf{a}}_r(\theta_1), \dots, \tilde{\mathbf{a}}_r(\theta_{L_s})]$,  $\mathbf{A}_t =[\tilde{\mathbf{a}}_t(\phi_1), \dots, \tilde{\mathbf{a}}_t(\phi_{L_s})]$, and $\mathbf{D} [k]= \text{diag}\left(\tilde{\alpha}_{1}[k], \dots, \tilde{\alpha}_{L_s}[k]\right)$. Note that in (\ref{eq::bicmb_su_channel_decompose}), $\mathbf{A}_t^{\perp}$ and $\mathbf{A}_r^{\perp}$ are unitary matrices spanning the null space of $\mathbf{H}$ and $\mathbf{H}^T$, $\widetilde{\mathbf{D}}=\left[\text{diag}\left(|\tilde{\alpha}_1|,\dots,|\tilde{\alpha}_{L_s}|\right) | \mathbf{0}_{(M_rN_r-L_s)\times (M_tN_t-L_s)}\right]$, and $\angle \tilde{\alpha}_{l} $ is the phase of the $l$th path.

In a similar way to \cite{Ayach2012,Sohrabi2017}, one can see that when the number of antennas at both the transmitter and the receiver goes to infinity, the optimal fully-digital precoders can be calculated by the hybrid precoding design in which $\mathbf{F}_\text{RF}= {\mathbf{A}}_t$ and $\mathbf{F}_{\text{BB}}[k] =  \boldsymbol{\Gamma}[k]$, where $\boldsymbol{\Gamma}[k]$ is the diagonal matrix of the allocated power for each data stream. One can use the water-filling approach to calculate elements of $\boldsymbol{\Gamma}[k]$. Similarly, for the combiner matrices at the receiver side, the asymptotic optimal analog combiner is given by the set of columns of $\mathbf{A}_r$ corresponding to $N_s$ largest complex gains, $\tilde{\alpha}_l$.

After employing the optimal beam at both the transmitter and the receiver side, the system input-output relation at $k$th subcarrier can be written as 
\begin{align}\label{eq::bicmb_ofdm_su_rec_svd}
y_{s}[k] = \sigma_s[k] x_s[k]  + n_s[k] 
\end{align}
for $s=1,\dots,N_s$ and $k=1,\dots,K$, where $\sigma_s$ is the $s$th largest singular value of $\mathbf{H}[k]$.

There are some criteria which we need to consider when designing the interleaver. For the BICMB-OFDM system, the interleaver is designed such that \cite{Akay2007, Sedighi2020J}
\begin{enumerate}
    \item Consecutive coded bits are interleaved within one OFDM symbol.
    \item Each subchannel should be used at least once within $d_{\text{free}}$ distinct bits among different codewords. 
    \item Consecutive coded bits need to be transmitted over different subcarriers of an OFDM symbol.
\end{enumerate}

%% file: diversity.tex
\section{Diversity Gain and PEP Analysis} \label{sec:diversity}

In this section, we carry out the analysis of the diversity gain for  BICMB-OFDM and we find the upper bound for the pairwise error probability (PEP).

\begin{theorem}
	Suppose that $N_r\rightarrow \infty$ and $N_t \rightarrow \infty$. Then the bit interleaved coded distributed massive MIMO system employing OFDM can achieve a diversity gain of
	\begin{align}\label{su_dg}
		G_d=\min (d_{\text{free}},M_rM_tL)
	\end{align}
\end{theorem}
\begin{proof}
	We model the BICMB bit interleaver as $\pi: k' \rightarrow (k,s,i)$, where $k'$ represents the original ordering of the coded bits $c_{k'}$, $k$ represents the time ordering of the signals $x_{k,s}$ and $i$ denotes the position of the bit $c_{k'}$ on symbol $x_{k,s}$.

We define $\chi_b^i$ as the subset of all signals $x\in\chi$. Note that the label has the value $b\in\{0,1\}$ in position $i$.

Then, the ML bit metrics are given by, \cite{Akay2007,Caire1998,Zehavi1992}
\begin{align}\label{ML_bit}
\gamma^i(y_{k,s},c_{k'})=\min_{x \in \chi^i_{c_{k'}}} \left|y_{s}[k]-\sigma_s[k]x\right|^2.
\end{align}

The receiver uses an ML decoder to make decisions based on
\begin{align}\label{ML_dec}
\mathbf{\underline{\hat{{\text{c}}}}}=\text{arg}\min_{\mathbf{\underline{\text{c}}}\in\mathcal{C}}\sum_{k'}\gamma^i(y_{s}[k],c_{k'}).
\end{align}

 Assume we transmit the code sequence $\underline{\text{c}}$ through the channel and we detect $\underline{\hat{\text{c}}}$. Then one can use (\ref{eq::bicmb_ofdm_su_rec_svd}) to write the PEP of $\underline{\text{c}}$ and $\underline{\hat{\text{c}}}$ as \ifCLASSOPTIONonecolumn
	\begin{align}
		P(\underline{\text{c}}\rightarrow \underline{\hat{\text{c}}}|\mathbf{H}[k],\forall k)
		=  P\left(\sum_{k'} \min_{x \in \chi^i_{c_{k'}}} |y_{s}[k]-\sigma_s[k]x|^2 \geq \sum_{k'} \min_{x \in \chi^i_{\hat{\text{c}}_{k'}}} |y_{s}[k]-\sigma_s[k]x|^2\right),
	\end{align} 
\else
	\begin{align}
		P&(\underline{\text{c}}\rightarrow \underline{\hat{\text{c}}}|\mathbf{H}[k],\forall k) \nonumber \\
		=  &P\left(\sum_{k'} \min_{x \in \chi^i_{c_{k'}}} \left|y_{s}[k]-\sigma_s[k]x\right|^2 \geq \sum_{k'} \min_{x \in \chi^i_{\hat{\text{c}}_{k'}}} \left|y_{s}[k]-\sigma_s[k]x\right|^2\right),
	\end{align} 
\fi where $s \in \{1,2,\dots,N_s\}$. 
		
Let's define
\begin{align} 
	&\tilde{x}_{s}[k]=\text{arg}\min_{x \in \chi^i_{{\text{c}}_{k'}}}\left|y_{s}[k]-\sigma_s[k]x\right|^2\label{eq::bicmb_ofdm_su_x_tilde}\\\label{eq::bicmb_ofdm_su_x_hat}
	&\hat{x}_{s}[k]=\text{arg}\min_{x \in \chi^i_{\bar{\text{c}}_{k'}}}\left|y_{s}[k]-\sigma_s[k]x\right|^2.
\end{align}

One can see that the two sets $\chi^i_{{\text{c}}_{k'}}$ and $\chi^i_{\bar{\text{c}}_{k'}} $ are complementary sets of constellation points in the signal constellation set $\chi$. Also, since $\tilde{x}_{k,s} \in \chi^i_{{\text{c}}_{k'}} $ and $\hat{x}_{k,s} \in \chi^i_{\bar{\text{c}}_{k'}} $, therefore $ \tilde{x}_{k,s} \neq \hat{x}_{k,s}$. Furthermore, whatever value we choose for ${x}_{k,s} $ from the set $\chi^i_{{\text{c}}_{k'}} $ the following inequality holds, since $\tilde{x}_{k,s}$ is the minimum based on (\ref{eq::bicmb_ofdm_su_x_tilde})
\begin{align}
	\left|y_{k,s}-\sigma_s x_{k,s}\right|^2 \geq \left|y_{k,s}-\sigma_s\tilde{x}_{k,s}\right|^2.
\end{align}

We want to guarantee there exist $d_{\text{free}}$ distinct pairs of $(\tilde{x}_s[k],\hat{x}_s[k])$ and $d_{\text{free}}$ distinct pairs of $({x}_s[k],\hat{x}_s[k])$  with $d_{\text{free}}$ distinct values of $k$. The interleaver design criteria 2 and 3 can make this happen. One can rewrite the PEP as 
\begin{align}
    P&(\underline{\text{c}}\rightarrow \underline{\hat{\text{c}}}|\mathbf{H}[k],\forall k) \nonumber\\
			=&  P\left(\sum_{k,d_{\text{free}}}  |y_{s}[k]-\sigma_s[k]x[k]|^2 \geq \sum_{k,d_{\text{free}}}  |y_{s}[k]-\sigma_s[k]x[k]|^2\right) \\
			& \leq P\left(\zeta \geq \sum_{k,d_{\text{free}}} \left| \sigma_s[k] \left(x_s[k]-\hat{x}_s[k]\right) \right|^2\right) \\ \label{eq::bicmb_ofdm_su_Q_function}
			& \leq Q\left(\sqrt{\frac{\sum_{k,d_{\text{free}}}d_{\text{min}}^2\sigma_s^2[k]}{2N_0}}\right),
			\end{align}
		where 
	$
		     \zeta = \sum_{k,d_{\text{free}}} \sigma_s[k](\hat{x}_{s}[k] - {x}_{s}[k])^*n_{s}[k] + \sigma_s[k](\hat{x}_{s}[k] - {x}_{s}[k])n_{s}^*[k]
	$
for given $\mathbf{H}[k], \forall k$, and $\zeta$ is a Gaussian random variable with zero mean and variance $2N_0\sum_{k,d_{\text{free}}}|\sigma_s[k]({x}_{s}[k] - \hat{x}_{s}[k])|^2$.

One can use an upper bound for the $Q$ function $Q(x) \leq \frac{1}{2}e^{-x^2/2}$ to find an upper bound for PEP in (\ref{eq::bicmb_ofdm_su_Q_function}) which is
\ifCLASSOPTIONonecolumn
\begin{align} \label{eq::bicmb_ofdm_su_PEP1}
    P(\underline{\text{c}}\rightarrow \underline{\hat{\text{c}}}) = \mathit{E}\left[ P(\underline{\text{c}}\rightarrow \underline{\hat{\text{c}}}|\mathbf{H}[k],\forall k)\right] 
			\leq \mathit{E} \left[ \frac{1}{2}\text{exp}\left(-\frac{d_{\text{min}}^2\sum_{k,d_{\text{free}}}\sigma_s^2[k]}{4N_0}\right)\right].
\end{align}
\else
\begin{align} \label{eq::bicmb_ofdm_su_PEP1}
P(\underline{\text{c}}\rightarrow \underline{\hat{\text{c}}}) &= \mathit{E}\left[ P(\underline{\text{c}}\rightarrow \underline{\hat{\text{c}}}|\mathbf{H}[k],\forall k)\right] \nonumber \\
&\leq \mathit{E} \left[ \frac{1}{2}\text{exp}\left(-\frac{d_{\text{min}}^2\sum_{k,d_{\text{free}}}\sigma_s^2[k]}{4N_0}\right)\right].
\end{align}
\fi
Similarly to \cite{Sedighi2020J}, one can see that the singular values for the channel matrix $\mathbf{H}[k]$ converge to $\alpha_{i,j}^l[k]$. And since by definition $\alpha_{ij}^l[k] \sim \mathcal{CN}(0,1)$, therefore $\sigma_s^2[k] \sim \chi_2^2$ has a chi-square distribution with two degrees of freedom. By defining $\omega_s[k]=\sigma_s^2[k]$, one can see that the probability density function (PDF) of $\omega_s[k]$ is $f(\omega_s[k])=\frac{1}{2\eta}\exp\left(-\frac{\omega_s[k]}{2 \eta}\right)$, where $\eta=\frac{N_rN_t}{L}$. Furthermore, $\alpha_s$ is defined as the number of the times the $s$th subchannel is used within $d_{\text{free}}$ bit, such that $\sum_{s=1}^{N_s}\alpha_s=d_{\text{free}}$, we can rewrite (\ref{eq::bicmb_ofdm_su_PEP1}) as
\ifCLASSOPTIONonecolumn
\begin{align}
     P(\underline{\text{c}}\rightarrow \underline{\hat{\text{c}}})
    & \leq  \frac{1}{2}\prod_{s=1}^{N_s} \left(\int_{0}^{\infty} \text{exp}\left( -\frac{  d_{\text{min}}^2 }  {4N_0} \omega_s[k]\right) f(\omega_{s}[k])d\omega_{s}[k]\right)^{\alpha_{s}} \\
    &= \frac{1}{2} \prod_{s=1}^{N_s} \left( \int_{0}^{\infty} \frac{1}{2\eta}\exp \left( -\frac{d_{min}^2}{4N_0} \omega_{s}[k]  -\frac{1}{2\eta}\omega_s[k] \right) d \omega_s[k]\right)^{\alpha_s}\\
    &= \frac{1}{2} \prod_{s=1}^{N_s} \left( \frac{1}{2\eta}  \frac{1}{ \frac{d_{\text{min}}^2}{4N_0} + \frac{1}{2\eta}} \right) ^{\alpha_s}
    \label{eq::bicmb_ofdm_su_PEP_alpha}\\
     & = \frac{1}{2} \prod_{s=1}^{N_s} \left(1 + \frac{d^2_{\text{min}} N_r }{4L} \text{SNR}\right) ^ {-\alpha_{s}}.
\end{align}
\else
\begin{align}
P&(\underline{\text{c}}\rightarrow \underline{\hat{\text{c}}})\nonumber\\
& \leq  \frac{1}{2}\prod_{s=1}^{N_s} \left(\int_{0}^{\infty} \text{exp}\left( -\frac{  d_{\text{min}}^2 }  {4N_0} \omega_s[k]\right) f(\omega_{s}[k])d\omega_{s}[k]\right)^{\alpha_{s}} \\
&= \frac{1}{2} \prod_{s=1}^{N_s} \left( \int_{0}^{\infty} \frac{1}{2\eta}\exp \left( -\frac{d_{min}^2}{4N_0} \omega_{s}[k]  -\frac{1}{2\eta}\omega_s[k] \right) d \omega_s[k]\right)^{\alpha_s}\\
&= \frac{1}{2} \prod_{s=1}^{N_s} \left( \frac{1}{2\eta}  \frac{1}{ \frac{d_{\text{min}}^2}{4N_0} + \frac{1}{2\eta}} \right) ^{\alpha_s}
\label{eq::bicmb_ofdm_su_PEP_alpha}\\
& = \frac{1}{2} \prod_{s=1}^{N_s} \left(1 + \frac{d^2_{\text{min}} N_r }{4L} \text{SNR}\right) ^ {-\alpha_{s}}.
\end{align}
\fi
We define the function $g(d,\bm{\alpha},\chi)$ as the PEP of two codewords with Hamming distance of $d$ from each other, where $\bm{\alpha}=[\alpha_{1},\dots,\alpha_{N_s}]$. Please note that the coefficient $\alpha_{s}$ depends on the codewords $\underline{\text{c}}$, $\underline{\hat{\text{c}}}$, $\text{d}(\underline{\text{c}}-\underline{\hat{\text{c}}})$, and the interleaver. In a similar way to 
\cite{Akay2007}, we can write down the union bound to illustrate the diversity order of the system. Therefore, one can calculate the BER $P_b$ as
\begin{align}\label{eq::bicmb_ofdm_delta}
    P_b \leq \frac{1}{k_c} \sum_{d=d_{\text{free}}}^{\infty} \sum_{i=1}^{W_I(d)} g(d,\bm{\alpha}(d,i),\chi).
\end{align}
Let us define the following
\begin{align} \label{eq::bicmb_ofdm_alpha_min}
    \Delta(\bm{\alpha}(d,i)) &= \sum_{s=1}^{N_s} \alpha_s(d,i) \\
    \bm{\alpha} (d_{\text{free}},j) &= \arg \min_{\substack{\bm{\alpha}(d_{\text{free}},i)\\ \label{eq::bicmb_ofdm_su_Delta_min} i=1,\dots,W_I(d_{\text{free}})}} \Delta(\bm{\alpha}(d_{\text{free}},i)).
\end{align}
Based on the definition in (\ref{eq::bicmb_ofdm_su_Delta_min}), we can see that the $\Delta(\bm{\alpha}(d,j))$ is the minimum value for $d\geq d_{\text{free}}$. Therefore, $\Delta(\bm{\alpha}(d,i))\geq \Delta(\bm{\alpha}(d,j))$ for $i=1,\dots,W_I(d)$. By using equations (\ref{eq::bicmb_ofdm_su_PEP_alpha})-(\ref{eq::bicmb_ofdm_su_Delta_min}), one can calculate  $P_b$ as 
\ifCLASSOPTIONonecolumn
\begin{align} \label{eq::bicmb_ofdm_su_BER_total}
    P_b \leq \frac{1}{k_c} \left[\sum_{d=d_{\text{free}}}^{\infty} \sum_{\substack{i=1i\neq j}}^{W_I(d)} g(d,\bm{\alpha}(d,i),\chi)\right. \; \;+ \left.\frac{1}{2} \prod_{s=1}^{N_s} \left(1+\frac{d_{\text{free}}^2 N_r}{4L}\text{SNR}\right) ^ {-\alpha_{s}( d_{\text{free}},j)}   \right].
\end{align}
\else
\begin{align} \label{eq::bicmb_ofdm_su_BER_total}
P_b \leq \frac{1}{k_c} &\left[\sum_{d=d_{\text{free}}}^{\infty} \sum_{\substack{i=1 i\neq j}}^{W_I(d)} g(d,\bm{\alpha}(d,i),\chi)\right. \nonumber \\&\; \;+ \left.\frac{1}{2} \prod_{s=1}^{N_s} \left(1+\frac{d_{\text{free}}^2 N_r}{4L}\text{SNR}\right) ^ {-\alpha_{s}( d_{\text{free}},j)}   \right].
\end{align}
\fi
Since we are working in the high SNR regime, we are only interested in the smallest power of SNR  in (\ref{eq::bicmb_ofdm_su_BER_total}). In high SNR regimes, the first term disappears. Therefore, the diversity gain, i.e., the power of SNR term would be $\Delta(\bm{\alpha}(d,j))$ when we transmit $N_s$ different symbols through the system.  Furthermore, BICMB-OFDM achieves full diversity order of $M_rM_tL$ when $M_rM_tL \leq \Delta(\bm{\alpha}(d,j))$ for spatial multiplexing order of $N_s$, which will be covered in the next section.
\end{proof}

%% file: simulation.tex
\section{Simulation Results} \label{sec:simulation}
In the simulations, we use the industry standard 64-state 1/2-rate (133,171) $d_{\text{free}}=10$ convolutional code. The large scale fading coefficient is assumed to be the same for all subchannels, i.e., $\beta = -20$ dB. For the sake of simplicity, only ULA array configuration with $d=0.5\lambda$ is considered at RAUs. Each OFDM symbol has $64$ subcarriers, and has $4$ $\mu s$ duration, of which $0.8$ $\mu s$ is CP. Information bits
are mapped onto 16 quadrature amplitude modulation ($16$-QAM) symbols in each subchannel. 
\ifCLASSOPTIONonecolumn
\begin{figure}[!t]
	\centering
	\includegraphics[width=.7\textwidth]{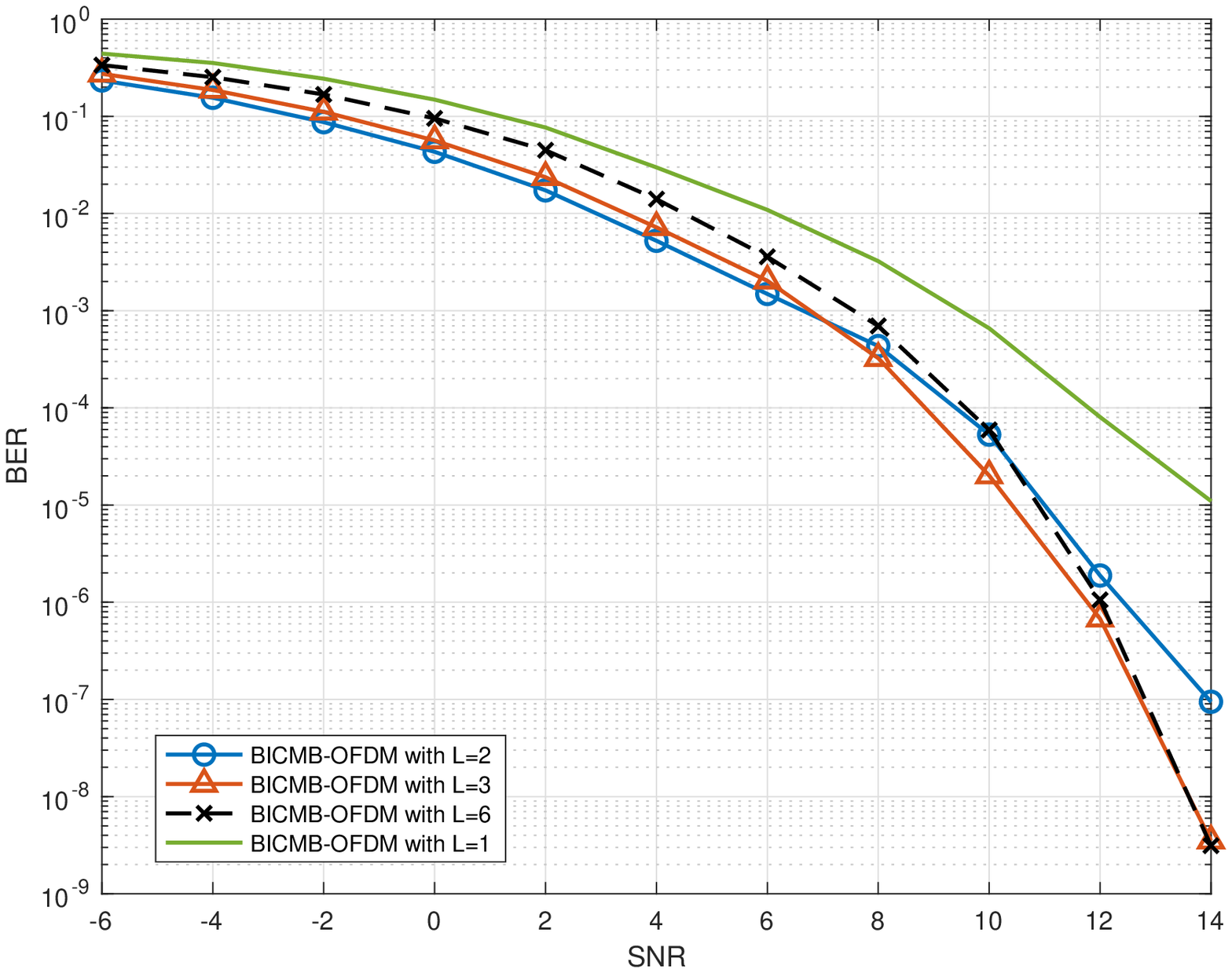}
	\caption{BER with respect to SNR for a $2\times 2$ with different values of $L$ and $N_s$.  }
	\label{fig::ofdm::2x2}
\end{figure}
\else
\begin{figure}[!t]
	\centering
	\includegraphics[width=.3\textwidth]{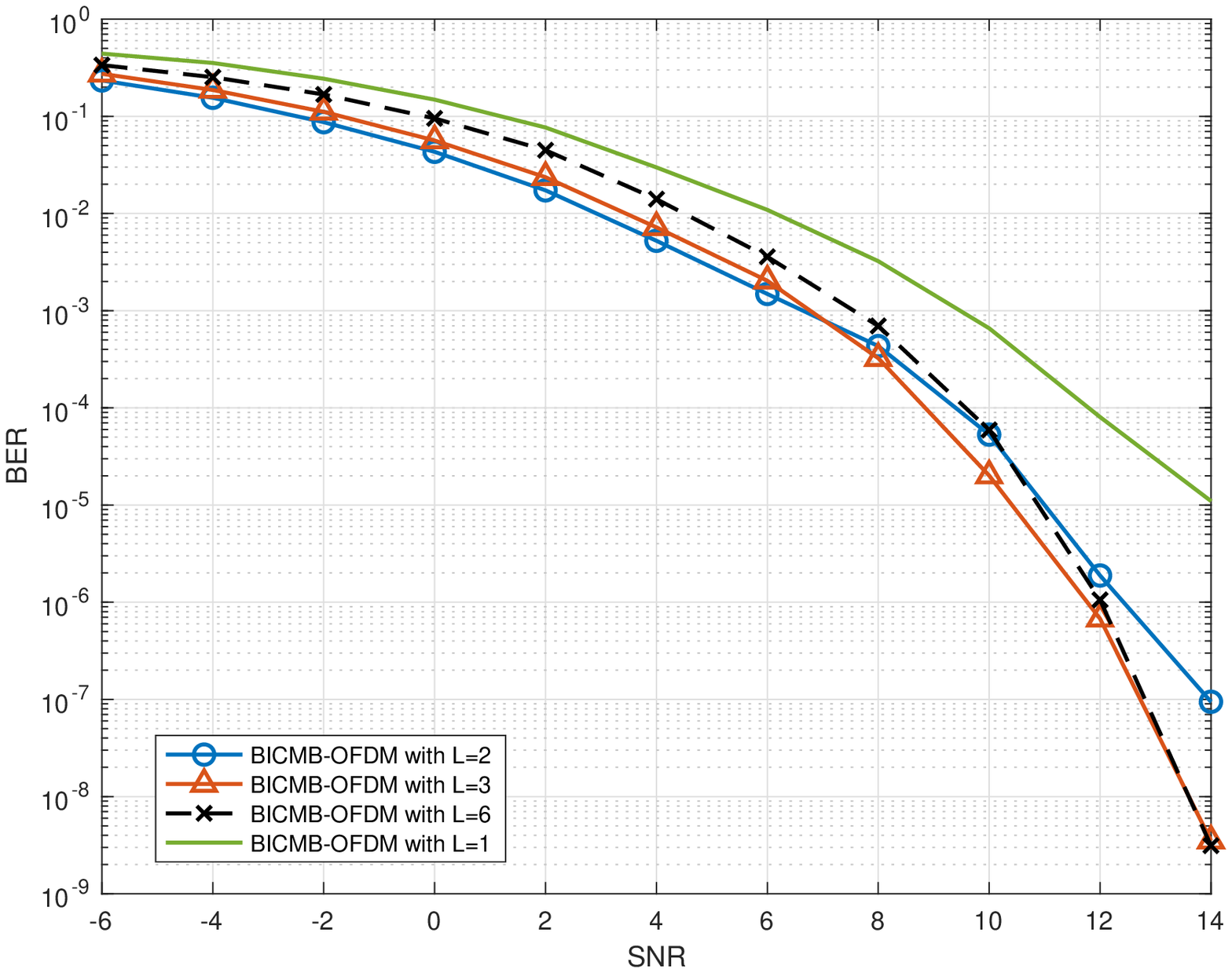}
	\caption{BER with respect to SNR for a $2\times 2$ with different values of $L$ and $N_s$.  }
	\label{fig::ofdm::2x2}
\end{figure}

\fi

Fig. \ref{fig::ofdm::2x2} illustrates the results for BICMB-OFDM for different number of propagation paths. In this figure, the number of RAUs at both the transmitter and the receiver side is two, i.e., $M_t=M_r=2$.  Since the spectrum of (133, 171) is used in the simulations for this paper, one can see that in this spectrum there are 11 codewords with a Hamming distance of $d_{\text{free}}$ from the all-zeros codeword. By comparing to the all-zeros codeword, the codeword $[1110010100010101110000000\dots]$ has the worst performance for BICMB-OFDM \cite{Akay2007}. On this codeword, the code and the interleaver combination result in the number of subchannel use vector $\bm{\alpha}$ to be $\bm{\alpha} = [3,7]$. This means that for the case of $L=1$ and $L=3$, the system achieves a diversity gain of $G_d=4$ and $G_d=8$, respectively. However, based on (\ref{eq::bicmb_ofdm_delta}), when the number of propagation paths increases, such that $M_rM_tL \geq \Delta(\bm{\alpha}(d,i))$, the maximum achievable diversity gain which is $\sum_{l=1}^{2}\alpha_l = 10$, is achieved. This result is confirmed by comparing two different cases, where $L=3$ and $L=6$ in the first and second case, respectively. Both of them achieve the same maximum diversity gain of $10$, which can be seen from Fig. {\ref{fig::ofdm::2x2}}.

Similarly, it can be seen from Fig. \ref{fig::ofdm::4x4} that the maximum diversity gain is achieved for the case when $M_r=M_t=4$. In this case, when compared to the all-zeros codeword, the codeword $[001110010100010101110000 \dots]$ leads to the worst diversity order \cite{Akay2007}. The $\bm{\alpha}$ coefficients are given as $\bm{\alpha}=[1,3,2,4]$ , which leads to a maximum diversity order of $10$  based on (\ref{eq::bicmb_ofdm_alpha_min}).

\ifCLASSOPTIONonecolumn
\begin{figure}[!t]
	\centering
	\includegraphics[width=.7\textwidth]{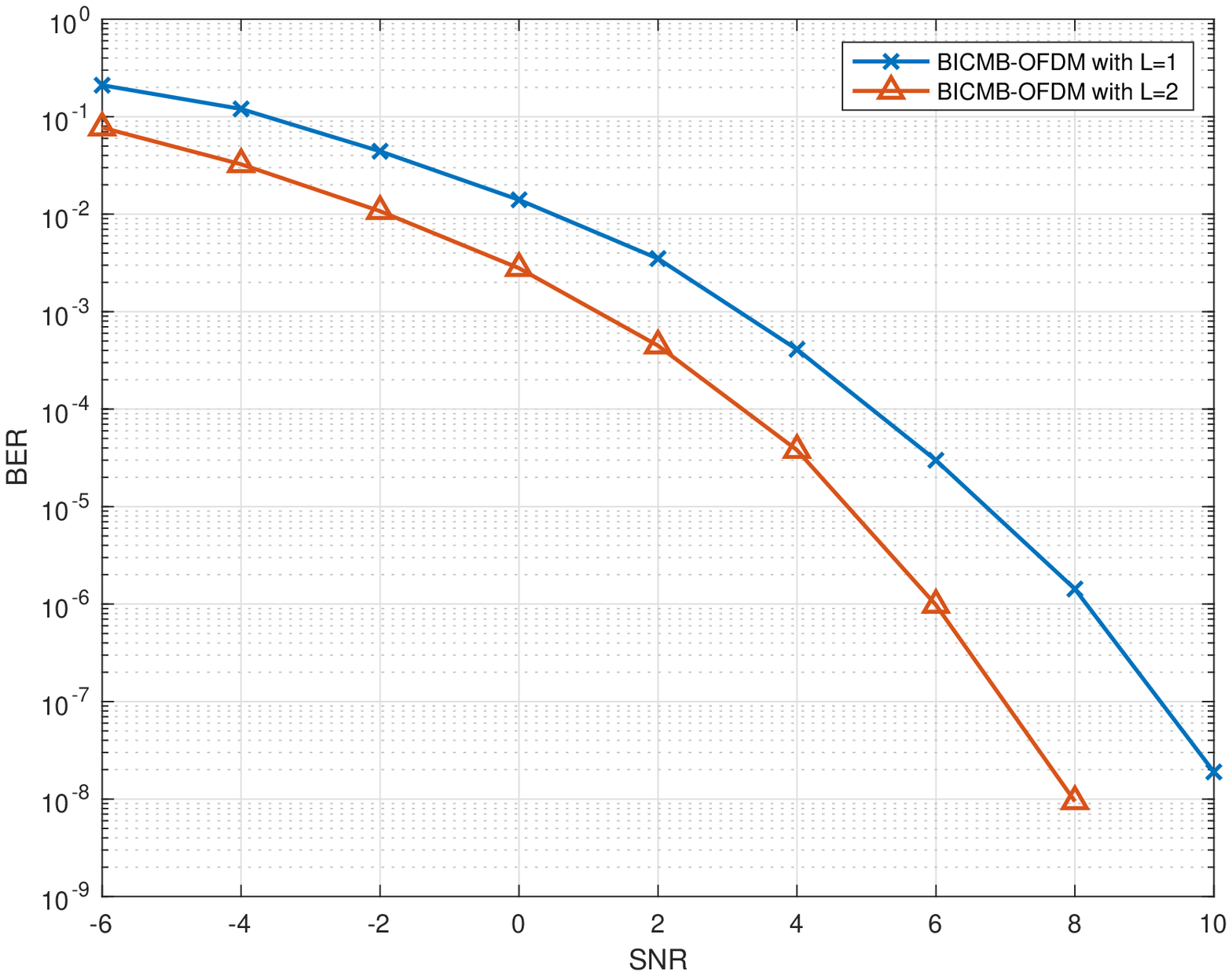}
	\caption{BER with respect to SNR for a $4\times 4$ with different values of $L$ and $N_s$.  }
	\label{fig::ofdm::4x4}
\end{figure}
\else
\begin{figure}[!t]
	\centering
	\includegraphics[width=.3\textwidth]{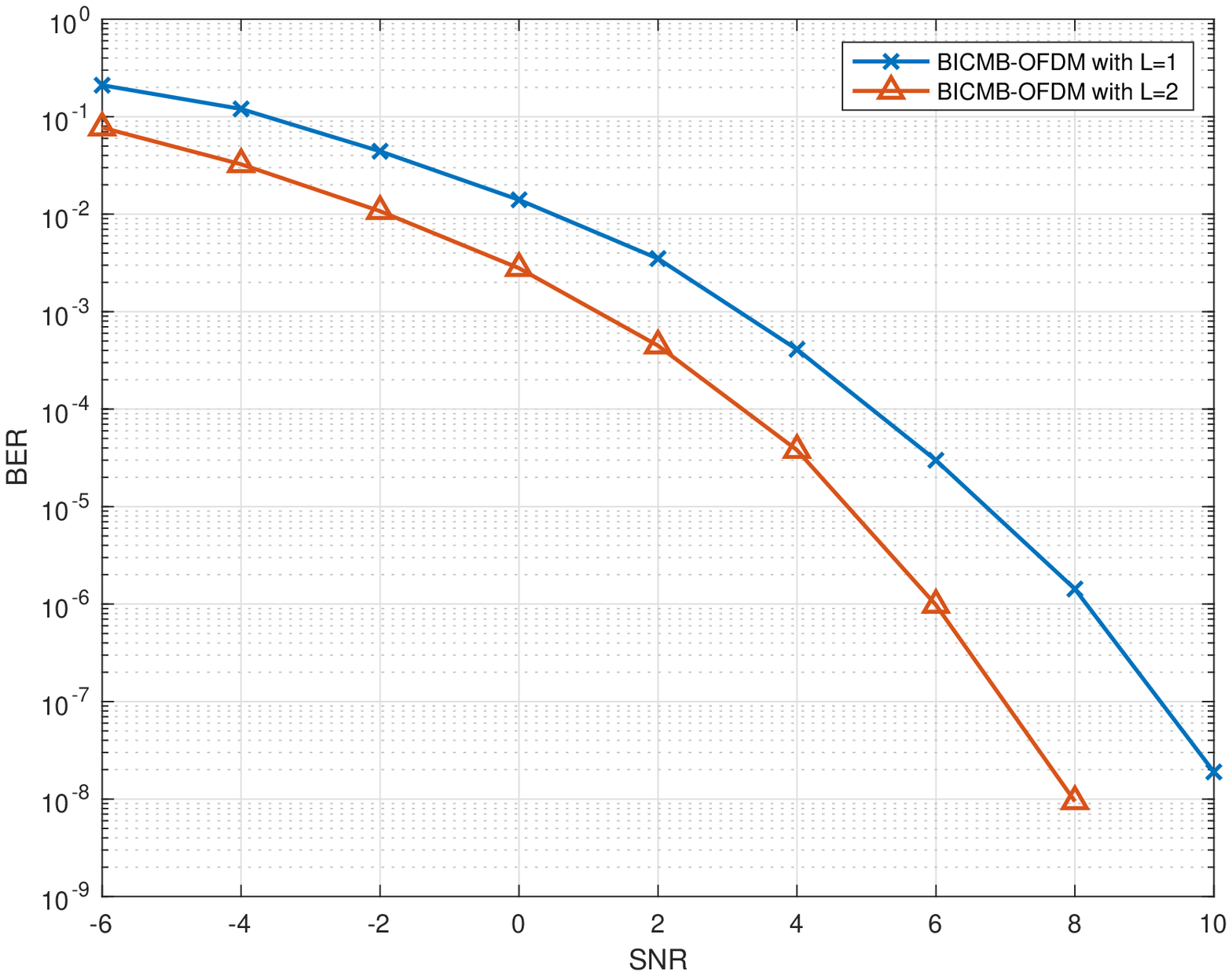}
	\caption{BER with respect to SNR for a $4\times 4$ with different values of $L$ and $N_s$.  }
	\label{fig::ofdm::4x4}
\end{figure}

\fi

Fig. \ref{fig::bicmb_ofdm_dif_antenna} illustrates the effect of the number of antennas at both the receiver side and the transmitter side on the diversity gain. One can see that changing the number of antennas at the RAUs does not affect the diversity gain. This confirms (\ref{eq::bicmb_ofdm_su_BER_total}) where the diversity gain is independent of the number of antennas at each RAU. Furthermore, one can see that by doubling the number of resources here, i.e., the number of antennas at the transmitter or the receiver, the performance of the system gets better by a factor of 3 dB.
\ifCLASSOPTIONonecolumn
\begin{figure}[t!]
	\centering
	\includegraphics[width=0.7\textwidth]{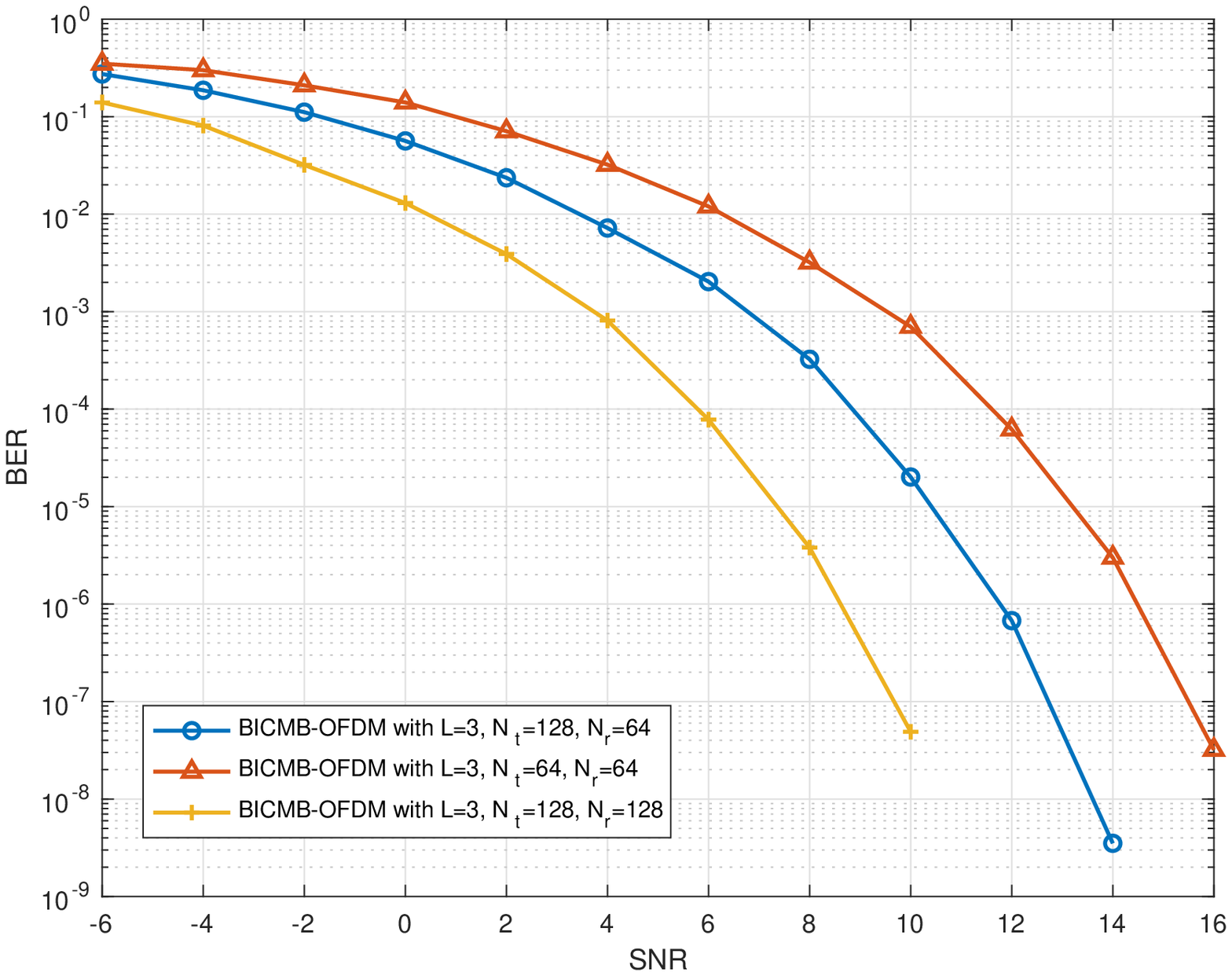}
	\caption{BER with respect to SNR for different number of antennas at each RAU at the transmitter and receiver. $M_r=M_t=2$.}
	\label{fig::bicmb_ofdm_dif_antenna}
\end{figure}
\else
\begin{figure}[t!]
	\centering
	\includegraphics[width=0.3\textwidth]{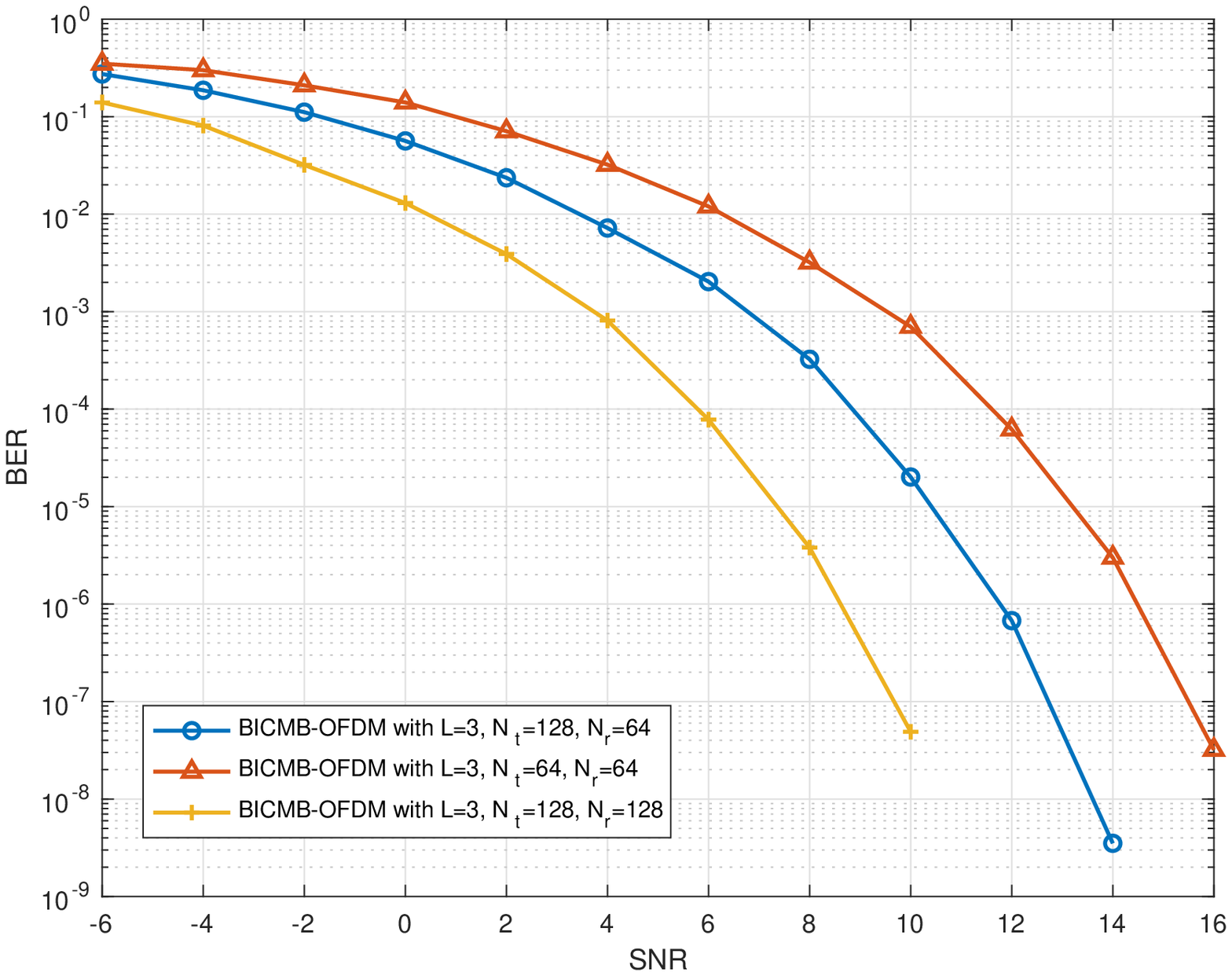}
	\caption{BER with respect to SNR for different number of antennas at each RAU at the transmitter and receiver. $M_r=M_t=2$.}
	\label{fig::bicmb_ofdm_dif_antenna}
\end{figure}

\fi

%% file: conclusion.tex
\section{Conclusion} \label{sec:conclusion}

In this work, we showed that by utilizing BICMB-OFDM in a mm-Wave  MIMO system with distributed antenna subarray architecture, one can achieve full diversity gain. This means, the diversity gain is independent of the number of transmitted data streams and can be increased by increasing the number of RAUs at the transmitter or the receiver. We also showed that the diversity gain is independent of the number of antennas at the RAUs in both the transmitter and the receiver.

%% file: main.bbl
\begin{thebibliography}{10}
\providecommand{\url}[1]{#1}
\csname url@samestyle\endcsname
\providecommand{\newblock}{\relax}
\providecommand{\bibinfo}[2]{#2}
\providecommand{\BIBentrySTDinterwordspacing}{\spaceskip=0pt\relax}
\providecommand{\BIBentryALTinterwordstretchfactor}{4}
\providecommand{\BIBentryALTinterwordspacing}{\spaceskip=\fontdimen2\font plus
\BIBentryALTinterwordstretchfactor\fontdimen3\font minus
  \fontdimen4\font\relax}
\providecommand{\BIBforeignlanguage}[2]{{%
\expandafter\ifx\csname l@#1\endcsname\relax
\typeout{** WARNING: IEEEtran.bst: No hyphenation pattern has been}%
\typeout{** loaded for the language `#1'. Using the pattern for}%
\typeout{** the default language instead.}%
\else
\language=\csname l@#1\endcsname
\fi
#2}}
\providecommand{\BIBdecl}{\relax}
\BIBdecl

\bibitem{Clark2001}
M.~V. {Clark}, T.~M. {Willis}, L.~J. {Greenstein}, A.~J. {Rustako}, V.~{Erceg},
  and R.~S. {Roman}, ``Distributed versus centralized antenna arrays in
  broadband wireless networks,'' in \emph{IEEE VTS 53rd Vehicular Technology
  Conference, Spring 2001. Proceedings (Cat. No.01CH37202)}, vol.~1, 2001, pp.
  33--37.

\bibitem{Roh2002}
W.~Roh and A.~{Paulraj}, ``{MIMO} channel capacity for the distributed
  antenna,'' in \emph{Proceedings IEEE 56th Vehicular Technology Conference},
  vol.~2, 2002, pp. 706--709 vol.2.

\bibitem{Dai2011}
L.~{Dai}, ``A comparative study on uplink sum capacity with co-located and
  distributed antennas,'' \emph{IEEE Journal on Selected Areas in
  Communications}, vol.~29, no.~6, pp. 1200--1213, 2011.

\bibitem{Wang2013}
D.~{Wang}, J.~{Wang}, X.~{You}, Y.~{Wang}, M.~{Chen}, and X.~{Hou}, ``Spectral
  efficiency of distributed {MIMO} systems,'' \emph{IEEE Journal on Selected
  Areas in Communications}, vol.~31, no.~10, pp. 2112--2127, 2013.

\bibitem{Qing2015}
\BIBentryALTinterwordspacing
Q.~Wang, D.~Debbarma, A.~Lo, Z.~Cao, I.~Niemegeers, and S.~Heemstra~de Groot,
  ``Distributed antenna system for mitigating shadowing effect in 60 {GHz
  WLAN},'' \emph{Wireless Personal Communications}, vol.~82, no.~2, pp.
  811--832, 2015. [Online]. Available:
  \url{https://doi.org/10.1007/s11277-014-2254-5}
\BIBentrySTDinterwordspacing

\bibitem{Juan2018}
\BIBentryALTinterwordspacing
J.~Cao, D.~Wang, J.~Li, Q.~Sun, and Y.~Hu, ``Uplink spectral efficiency
  analysis of multi-cell multi-user massive {MIMO} over correlated {Ricean}
  channel,'' \emph{Science China Information Sciences}, vol.~61, no.~8, p.
  082305, 2018. [Online]. Available:
  \url{https://doi.org/10.1007/s11432-017-9278-0}
\BIBentrySTDinterwordspacing

\bibitem{Dian2018J}
D.~Yu, S.~Xu, and H.~H. Nguyen, ``Diversity gain of millimeter-wave massive
  {MIMO} systems with distributed antenna arrays,'' \emph{EURASIP Journal on
  Wireless Communications and Networking}, vol.~54, pp. 1--13, 2019.

\bibitem{Sedighi2020J}
S.~Sedighi and E.~Ayanoglu, ``Bit-interleaved coded multiple beamforming in
  millimeter-wave massive {MIMO} systems,'' \emph{IEEE Transactions on
  Communications}, vol.~68, no.~10, pp. 6174--6185, 2020.

\bibitem{Sedighi2020L}
------, ``Bit-interleaved coded multiple beamforming with perfect coding in
  millimeter-wave {MIMO} systems,'' \emph{IEEE Wireless Communications
  Letters}, vol.~10, no.~3, pp. 644--648, 2021.

\bibitem{Xiao2015}
Z.~{Xiao}, X.~{Xia}, D.~{Jin}, and N.~{Ge}, ``Iterative eigenvalue
  decomposition and multipath-grouping {Tx/Rx} joint beamformings for
  millimeter-wave communications,'' \emph{IEEE Transactions on Wireless
  Communications}, vol.~14, no.~3, pp. 1595--1607, March 2015.

\bibitem{Elganimi2018}
T.~Y. {Elganimi} and A.~A. {Elghariani}, ``Space-time block coded spatial
  modulation aided millimeter-wave {MIMO} with hybrid precoding,'' in
  \emph{2018 26th Signal Processing and Communications Applications Conference
  (SIU)}, 2018, pp. 1--6.

\bibitem{Pi2011}
Z.~{Pi} and F.~{Khan}, ``An introduction to millimeter-wave mobile broadband
  systems,'' \emph{IEEE Communications Magazine}, vol.~49, no.~6, pp. 101--107,
  2011.

\bibitem{Zhang2020}
\BIBentryALTinterwordspacing
Y.~Zhang, D.~Wang, Y.~Huo, X.~Dong, and X.~You, ``Hybrid beamforming design for
  mmwave {OFDM} distributed antenna systems,'' \emph{Science China Information
  Sciences}, vol.~63, no.~9, p. 192301, 2020. [Online]. Available:
  \url{https://doi.org/10.1007/s11432-019-2799-y}
\BIBentrySTDinterwordspacing

\bibitem{Sohrabi2017}
F.~Sohrabi and W.~Yu, ``Hybrid analog and digital beamforming for mmwave {OFDM}
  large-scale antenna arrays,'' \emph{IEEE Journal on Selected Areas in
  Communications}, vol.~35, no.~7, pp. 1432--1443, 2017.

\bibitem{Akay2007}
E.~Akay, E.~Sengul, and E.~Ayanoglu, ``Bit-interleaved coded multiple
  beamforming,'' \emph{IEEE Trans. Commun.}, vol.~55, no.~9, pp. 1802--1811,
  Sep. 2007.

\bibitem{Li2013}
B.~{Li} and E.~{Ayanoglu}, ``Diversity analysis of bit-interleaved coded
  multiple beamforming with orthogonal frequency division multiplexing,''
  \emph{IEEE Transactions on Communications}, vol.~61, no.~9, pp. 3794--3805,
  2013.

\bibitem{Ayach2012}
O.~E. {Ayach}, R.~W. {Heath}, S.~{Abu-Surra}, S.~{Rajagopal}, and Z.~{Pi},
  ``The capacity optimality of beam steering in large millimeter-wave {MIMO}
  systems,'' \emph{2012 IEEE 13th International Workshop on Signal Processing
  Advances in Wireless Communications (SPAWC)}, pp. 100--104, Jun. 2012.

\bibitem{Caire1998}
G.~Caire, G.~Taricco, and E.~Biglieri, ``Bit-interleaved coded modulation,''
  \emph{IEEE Trans. Inf. Theory}, vol.~44, no.~3, pp. 927--946, May1998.

\bibitem{Zehavi1992}
E.~Zehavi, ``8-{PSK} trellis codes for a {Rayleigh} channel,'' \emph{IEEE
  Trans. Commun.}, vol.~40, no.~5, pp. 873--884, May 1992.

\end{thebibliography}
